\documentclass[letterpaper]{article} %
\usepackage{aaai24}  %
\usepackage{times}  %
\usepackage{helvet}  %
\usepackage{courier}  %
\usepackage[hyphens]{url}  %
\usepackage{graphicx} %
\usepackage{natbib}  %
\usepackage{caption} %
\usepackage{algorithm}

\usepackage{newfloat}
\usepackage{listings}
\DeclareCaptionStyle{ruled}{labelfont=normalfont,labelsep=colon,strut=off} %
\floatstyle{ruled}
\newfloat{listing}{tb}{lst}{}
\floatname{listing}{Listing}
\usepackage{amsmath}
\usepackage{amsthm}
\usepackage{booktabs}
\usepackage{algpseudocode}
\algnewcommand\algorithmicforeach{\textbf{for each}}
\algdef{S}[FOR]{ForEach}[1]{\algorithmicforeach\ #1\ \algorithmicdo}
\usepackage{flexisym}
\usepackage{xspace}
\usepackage{cleveref}
\usepackage{xcolor}
\usepackage{booktabs}
\usepackage{subcaption}
\usepackage{array}
\usepackage{todonotes}
\usepackage{multirow}
\algnotext{EndFor}
\algnotext{EndIf}
\newcommand{\toolname}{\ensuremath{\mathsf{sharpASP}}}

\title{Exact ASP Counting with Compact Encodings}

\author {
    Mohimenul Kabir\textsuperscript{\rm 1},
    Supratik Chakraborty\textsuperscript{\rm 2},
    Kuldeep S Meel\textsuperscript{\rm 3}
}
\affiliations {
    \textsuperscript{\rm 1}National University of Singapore\\
    \textsuperscript{\rm 2}Indian Institute of Technology Bombay\\
    \textsuperscript{\rm 3}University of Toronto\\
}

\usepackage{bibentry}

\newcommand{\hide}[1]{}

\definecolor{supratikcolor}{HTML}{137403}
\definecolor{kuldeepcolor}{rgb}{0.2,0.6,0.6}
\definecolor{mahicolor}{rgb}{0,0.4,0.8}
\definecolor{todocolor}{rgb}{0.9,0.1,0.1}
\definecolor{changedcolor}{rgb}{0.42,0.27,0.57}
\definecolor{removedcolor}{rgb}{0.867,0.176,0.361}
\definecolor{reorganisedcolor}{rgb}{0.553,0.169,0.565}
\definecolor{newcolor}{rgb}{0.82,0.286,0.043}

\newcommand{\Card}[1]{|#1|}
\newcommand{\var}[1]{\mathsf{Var}(#1)}
\newcommand{\copyvar}[1]{\mathsf{CopyVar}(#1)}
\newcommand{\copyop}[1]{\mathsf{Copy}(#1)}
\newcommand{\decompose}[1]{\mathsf{Decomposition}(#1)}

\newcommand{\loopatoms}[1]{\mathsf{LA}(#1)}
\newcommand{\completion}[1]{\mathsf{Comp}(#1)}
\newcommand{\loopformula}[1]{\mathsf{LF}(#1)}
\newcommand{\loopa}[1]{\mathsf{Loops}(#1)}
\newcommand{\stable}[1]{\mathsf{AS}(#1)}

\newcommand{\decide}[1]{\mathsf{PickNonCopyVar}(#1)}

\newcommand{\external}[1]{\mathsf{ExtRule}(#1)}
\newcommand{\answer}[2]{\mathsf{AS}(#1,#2)}
\newcommand{\cnt}{\mathsf{Count}}
\newcommand{\cache}{\mathsf{Cache}}
\newcommand{\countsm}[2]{\mathsf{CntAS}(#1,#2)}
\newcommand{\countAS}[1]{\mathsf{CntAS}(#1)}
\newcommand{\copyatom}[1]{#1\textprime}
\newcommand{\true}{$\mathsf{true}$\xspace}
\newcommand{\false}{$\mathsf{false}$\xspace}

\newcommand{\ganak}{GANAK\xspace}
\newcommand{\dfour}{D4\xspace}

\newcommand{\sharpsattd}{SharpSAT-TD\xspace}

\newcommand{\asptosharpsat}{lp2sat+\shortsharpsattd\xspace}

\newcommand{\aspmcsharpsat}{aspmc+\shortsharpsattd\xspace}

\newcommand{\sharpasp}{\ensuremath{\mathsf{sharpASP}\xspace}}
\newcommand{\asproblog}{ASProblog}
\newcommand{\asproblogshort}{ASProb}

\newcommand{\shortganak}{G}
\newcommand{\shortsharpsattd}{STD}

\newcommand{\counter}[2]{\mathsf{Counter}(#1,#2)}
\newcommand{\answers}[1]{\Card{\mathsf{AS}(#1)}}

\newcommand{\body}[1]{\mathsf{Body}(#1)}
\newcommand{\at}[1]{\mathsf{atoms}(#1)}
\newcommand{\head}[1]{\mathsf{Head}(#1)}
\newcommand{\rules}[1]{\mathsf{Rules}(#1)}
\newcommand{\dependency}[1]{\mathsf{DG}(#1)}

\newcommand{\SB}{\{\,}%
\newcommand{\SE}{\,\}}%

\newtheorem{proposition}{Proposition}
\newtheorem{definition}{Definition}
\newtheorem{theorem}{Theorem}
\newtheorem{example}{Example}

\newcommand{\clingo}{clingo\xspace}
{%
  \endgroup
}%
{%
  \addtocounter{theorem}{-1}
  \endgroup
}%

\begin{document}

\maketitle

\begin{abstract}
    Answer Set Programming (ASP) has emerged as a promising paradigm
    in knowledge representation and automated reasoning owing to its
    ability to model hard combinatorial problems from diverse domains
    in a natural way.  Building on advances in propositional SAT
    solving, the past two decades have witnessed the emergence of
    well-engineered systems for solving the answer set satisfiability
    problem, i.e., finding models or answer sets for a given answer
    set program.  In recent years, there has been growing interest in
    problems beyond satisfiability, such as model counting, in the
    context of ASP.
    Akin to the early days of propositional model counting,
    state-of-the-art exact answer set counters do not scale well
    beyond small instances. 
    Exact ASP counters struggle with handling larger input formulas.
    The primary contribution of this paper is a new ASP counting
    framework, called {\toolname}, which counts answer sets avoiding
    larger input formulas.  This relies on an alternative way of
    defining answer sets that allows for the lifting of key techniques
    developed in the context of propositional model counting.  Our
    extensive empirical analysis over $1470$ benchmarks demonstrates
    significant performance gain over current state-of-the-art exact
    answer set counters.  Specifically, %
    by using {\toolname}, we were able to solve $1062$ benchmarks with
    PAR2 score of $3082$ whereas using prior state-of-the-art, we could
    only solve $895$ benchmarks with a PAR2 score of $4205$, all other
    experimental conditions being the same.
\end{abstract}
\section{Introduction}
Answer Set Programming (ASP) \cite{MT1999} is a declarative 
problem-solving approach with a wide variety of applications 
ranging from planning, diagnosis, scheduling, and product
configuration checking \cite{NYEP2016,BR2015,TSNS2003}.
An ASP program consists of a set of rules defined over
propositional atoms, where each rule logically expresses an implication 
relation. An assignment to the propositional atoms satisfying the ASP semantic is called an \emph{answer set}. 
In this paper, we focus on an important class of
ASP programs called \textit{normal logic programs} that have been used
in diverse applications (see for example~\cite{DM2017,BEEMR2007}), and
present a new technique to count answer sets of such programs, while
scaling much beyond state-of-the-art exact answer set counters.

In general, given a set of constraints in a theory, model counting
seeks to determine the number of models (or solutions) to the set of
constraints. From a computational complexity perspective, this can be
significantly harder than deciding whether there exists any solution
to the set of constraints, i.e. the satisfiability problem.  Yet, in
the context of propositional reasoning, compelling applications have
fuelled significant practical advances in propositional model
counting, also referred to as \#SAT~\cite{SRSM2019,Thurley2006}, over
the past decade.
This, in turn, has
ushered in new applications in quantified information
flow~\cite{BEHLMQ2018}, neural network verification~\cite{BSSMS2019}, computational biology,
and the like. The success of practical propositional
model counting in diverse application domains
have naturally led researchers to ask if practically efficient
counting algorithms can be devised for constraints beyond
propositional logic.
In particular, 
there has been growing interest in answer set
counting,
motivated by applications in probabilistic reasoning and
network reliability~\cite{KM2023}.

Early efforts to build answer set counters
sought to work by enumerating answer sets of a given
ASP program~\cite{FHMW2017,GKNS2007}.
While this works extremely well for answer set counts
upto a certain threshold, enumeration doesn't scale well
for problem instances with too many answer sets.
Therefore, subsequent approaches to answer set counting sought to leverage
the significant progress made in \#SAT techniques.  Specifically, Aziz
et al.~\cite{ACMS2015} integrated a \textit{component-caching} based
propositional model counting technique with \textit{unfounded set detection} 
to yield an answer set counter, called \asproblog.  In another line of work, dynamic
programming on a tree decomposition of the input problem instance has
been proposed to achieve scalability for ASP instances with low
treewidth~\cite{FHMW2017,FH2019}.  Yet another approach has been to
translate a given normal logic program $P$ into a propositional
formula $F$, such that there is a one-to-one correspondence between
answer sets of $P$ and models of
$F$~\cite{JN2011,Bomanson2017,Janhunen2006}. Answer sets of $P$
can then be counted by invoking an off-the-shelf propositional model
counter~\cite{SRSM2019} on $F$.  Though promising in principle, a naive application of
this approach doesn't scale well in practice owing to a blowup in the
size of the resulting formula $F$ when the implications between propositional atoms encoded in
  the program $P$ give rise to circular dependencies~\cite{LR2006}, which is a common occurrence when modeling 
  numerous real-world applications.  To address this,
  researchers have proposed techniques to transform the program to
  effectively break such circular dependencies and then use a
  treewidth-aware translation of the transformed program to a
  propositional formula (see, for example~\cite{EHK2021}).  However,
  breaking such circular dependencies can increase the treewidth of
  the resulting transformed problem instance, which in turn can
  adversely affect the performance of answer set counting.
Thus, despite significant advances, state-of-the-art exact answer set counters
are stymied by scalability bottlenecks, limiting their practical
applicability. Within this context, we ask the question: {\em Can we design a scalable answer set counter, 
accompanied by a substantial reduction in the size of the transformed input program, 
particularly when addressing circular dependencies?}

The principal contribution of this paper addresses the aforementioned question by introducing an alternative approach 
to exact answer set counting, called
{\toolname}, while alleviating key bottlenecks faced by earlier
approaches.  While a mere reduction in translation size does not inherently establish a scalable ASP counting solution 
for general scenarios, {\toolname} allows us to solve
larger and more instances of exact answer set counting than was
feasible earlier.  Similar to \asproblog, {\toolname} lifts component-caching based propositional model counting
algorithms to ASP counting. The key idea that makes this possible
is an alternative yet correlated perspective on defining answer sets.  This alternative definition makes it
possible to lift core ideas like decomposability and determinism in
propositional model counters to facilitate answer set counting.  Viewed
differently, transforming propositional model counters into our proposed ASP counting framework requires minimal adjustments. Our
experimental analysis demonstrates that {\toolname}, 
built using this approach, significantly outperforms the performance of previous 
state-of-the-art techniques across instances from diverse domains.
This serves to underscore the effectiveness of our approach over the combined might of earlier state-of-the-art exact
answer set counters.

The remainder of this paper is organized as follows. We present some
preliminaries and notations in~\Cref{sec:preliminaries}.
\Cref{sec:algorithm}
presents an alternative way of defining the answer set of an ASP
instance, which allows us to propose the answer set counting algorithm 
of {\toolname} in~\Cref{sec:implemetation}, where we also
present correctness arguments for our algorithm.
\Cref{sec:experiment} presents our experimental evaluation of the proposed 
answer set counting algorithm. Finally, we conclude our paper in~\Cref{sec:conclusion}.

\section{Preliminaries}
\label{sec:preliminaries}
Before delving into the details, we introduce some notation and
preliminaries from propositional satisfiability and answer set
programming.
\paragraph{Propositional Satisfiability.}
A propositional \emph{variable} $v$ takes one of two values: $0$
(denoting $\mathsf{false}$) or $1$ (denoting $\mathsf{true}$).  A
\emph{literal} $\ell$ is either a variable (positive literal) or its
negation (negated literal),
 and a \emph{clause} $C$ is a disjunction of literals.  For
 convenience of exposition, we sometimes represent a clause as a set of
 literals, with the implicit understanding that all literals in the
 set are disjoined in the clause.  A clause with a single literal
 is also called a \emph{unit clause}.
In general, the constraint represented by a clause $C \equiv
(\neg{v_1} \vee \ldots \vee \neg{v_k} \vee v_{k+1} \vee \ldots \vee
v_{k+m})$ can be expressed as a logical \textit{implication}: $(v_1
\wedge \ldots \wedge v_k) \rightarrow (v_{k+1} \vee \ldots \vee
v_{k+m})$.  If $k = 0$, the antecedent of the above implication is
$\mathsf{true}$, and if $m = 0$, the consequent is $\mathsf{false}$.
A \emph{conjunctive normal form (CNF)} formula $\phi$ is a conjuction
of clauses.  When there is no confusion, a CNF formula is also
sometimes represented as a set of clauses, with the implicit
understanding that all clauses in the set are conjoined to give the
formula.  We denote the set of variables in $\phi$ as $\var{\phi}$.

An assignment of a set $X$ of propositional variables is a mapping
$\tau: X \rightarrow \{0,1\}$.  For a variable $v \in X$, we define
$\tau(\neg{v}) = 1 - \tau(v)$.
Given a CNF formula $\phi$ (as a set of clauses) and an assignment
$\tau: X \rightarrow \{0,1\}$, where $X \subseteq \var{\phi}$,
the \textit{unit propagation} of $\tau$ on $\phi$, denoted
$\phi|_{\tau}$, is recursively defined as follows:

$\phi|_{\tau} = \begin{cases}
    1 & \text{if $\phi \equiv 1$}\\
    \phi'|_{\tau} & \text{if $\exists C \in \phi$ s.t. $\phi' = \phi \setminus \{C\}$,}\\ & \text{$\ell \in C$ 
    and $\tau(\ell) = 1$} \\
    \phi'|_{\tau} ~\cup~ \{C'\}  & \text{if $\exists C \in \phi$ s.t. $\phi' = \phi' \setminus \{C\}$,}\\ & \text{$\neg \ell \in C$,}  \text{$C' = C \setminus \{\neg \ell\}$} \\ &
       \text{ and ($\tau(\ell) = 1$ or $\{l\} \in \phi$)}
\end{cases}$

\noindent
Note that $\phi|_{\tau}$ always reaches a fixpoint. We say that $\tau$ \emph{unit propagates} to literal $\ell$ in $\phi$
if $\{\ell\} \in \phi|_{\tau}$, i.e. if $\phi|_{\tau}$ has a
unit clause with the literal $\ell$.

\paragraph{Answer Set Programming.}
An answer set program $P$ expresses logical constraints
between a set of propositional variables.  In the context of answer
set programming, such variables are also called \emph{atoms}, and the
set of atoms appearing in $P$ is denoted $\at{P}$.  For notational
convenience, we will henceforth use the terms ``variable'' and
``atom'' interchangeably.  A \textit{normal (logic) program} is a set
of rules of the following form:
\begin{align}
\label{eq:basic_rule}
   \text{Rule $r$:~~}  a \leftarrow b_1, \ldots, b_m, \sim c_1, \ldots,  \sim c_n
\end{align}

In the above rule, $\sim$ denotes \textit{default negation}, signifying \textit{failure to
  prove}~\cite{clark1978}.  For rule $r$ shown above, atom ``$a$'' is
called the \textit{head of $r$} and is denoted $\head{r}$.  Similarly,
the set of literals $\{b_1, \ldots, b_m, \sim c_1, \ldots, \sim c_n\}$ is
called the \textit{body} of $r$.  Specifically, $\{b_1, \ldots, b_m\}$
are the \emph{positive body atoms}, denoted $\body{r}^+$, and $\{c_1,
\ldots, c_n\}$ are the \emph{negative body atoms}, denoted
$\body{r}^-$.  For purposes of the following discussion, we use
$\body{r}$ to denote the conjunction $b_1 \wedge \ldots \wedge b_m
\wedge \neg{c_1} \wedge \ldots \wedge \neg{c_n}$.  Atoms that appear
in the head of a rule (like $a$ in rule $r$ above) have also been
called \emph{founded variables/atoms} in the
literature~\cite{ACMS2015}.

In answer set programming, an interpretation $M \subseteq \at{P}$
lists the $\mathsf{true}$ atoms, i.e., an atom $a$ is $\mathsf{true}$
iff $a \in M$.
An assignment $M$ satisfies $\body{r}$, denoted $M \models
\body{r}$, iff $\body{r}^+ \subseteq M$ and $\body{r}^- \cap M =
\emptyset$, where $\sim$ is interpreted classically, i.e., $M \models
\sim c_i$ iff $M \not\models c_i$.
The rule $r$ (see \Cref{eq:basic_rule}) specifies that if all atoms in
$\body{r}^+$ \textit{hold} and no atom in $\body{r}^-$ holds, then
$\head{r}$ also holds.  The assignment $M$ satisfies rule
$r$, denoted $M \models r$, if and only if whenever $M \models \body{r}$,
then $\head{r} \in M$.  Let $\rules{P}$ denote the set of
all rules in a normal program $P$.  Then, we say that an
assignment $M$ satisfies $P$, denoted $M \models P$, if and only
if $M \models r$ for each $r \in \rules{P}$.

Given an assignment (or set of atoms) $M$, the
\textit{Gelfond-Lifschitz (GL) reduct} of a program $P$ w.r.t. $M$ is defined as
$P^M = \{\head{r} \leftarrow \body{r}^+~|~ r \in \rules{P}, \body{r}^-
\cap M = \emptyset\}$~\cite{GL1988}.
A set of atoms $M$ is an answer set of $P$ if and only if $M \models
P^M$, but $N \not\models P^M$ for every proper subset $N$ of $M$.
The set of all answer sets of program $P$ is denoted by $\stable{P}$, and the
answer set counting problem is to compute $\answers{P}$, which is denoted by $\countAS{P}$.

\textit{Clark’s completion} \cite{clark1978} or \emph{program
completion} is a technique for obtaining a translation of a normal
program $P$ into a related, but not semantically equivalent,
propositional formula $\completion{P}$.  Specifically, for each atom
$a \in \at{P}$, we do the following:
\begin{enumerate}
    \item Let $r_1, \ldots, r_k \in \rules{P}$ such
      that $\head{r_1} =\ldots= \head{r_k} = a$, then we add the
      propositional formula $(a \leftrightarrow (\body{r_1} \vee
      \ldots \vee \body{r_k}))$ to $\completion{P}$.
    \item Otherwise, we add the literal $\neg{a}$ to $\completion{P}$.
\end{enumerate}
Finally, $\completion{P}$ is obtained as the logical conjunction of
all constraints added above.  It has been shown in the literature that
an answer set of $P$ satisfies $\completion{P}$ but not vice
versa~\cite{LZ2004}.

To overcome the above problem, the idea of \textit{loop formula} was
introduced in~\cite{LZ2004}.  We outline the construction of a loop
formula below.  Given a normal program $P$, we start by defining the
\textit{positive dependency graph} $\dependency{P}$ of $P$ as follows.
The vertices of $\dependency{P}$ are simply $\at{P}$.  For $a, b \in
\at{P}$, there exists an edge from $b$ to $a$ in $\dependency{P}$ if
there is a rule $r \in \rules{P}$ such that $a \in \body{r}$ and $b =
\head{r}$.
A set of atoms $L \subseteq \at{P}$ constitutes a \textit{loop} in $P$
if for every two atoms $x,y \in L$ there is a path from $x$ to $y$ in
$\dependency{P}$ such that all atoms (nodes) on the path are in
$L$.
An atom $a$ is called a \textit{loop atom} of $P$ if there is a loop
$L$ in $P$ such that $a \in L$.
We use $\loopa{P}$ and $\loopatoms{P}$ to denote the set of all loops
and the set of all loop atoms of $P$, respectively.
A program $P$ is called \textit{tight} if there is no loop in $P$; otherwise,  
$P$ is called \textit{non-tight}.
Lin and Zhao~\cite{LZ2004} showed that atoms in a loop cannot be
asserted $\mathsf{true}$ by themselves; instead they must be asserted
by some atoms external to the loop. %
Specifically, a rule $r$ is an \textit{external support} of a loop $L$
in $P$ if $\head{r} \in L$ and $\body{r}^+ \cap L = \emptyset$.
Let $\external{L}$ denote the set of all external supports of loop $L$
in $P$.
The loop formula $\loopformula{L,P}$~\cite{LL2003} of a loop $L$ in program $P$ can now be defined as follows:
\begin{align*}
    \loopformula{L,P} &= \big(\bigwedge_{a \in L} a\big) \rightarrow \bigvee_{r \in \external{L}} \body{r}
\end{align*}
Finally, the loop formula $\loopformula{P}$ of program $P$ is defined as the conjunction of loop formulas for all loops $L$ in $P$, i.e.
    $\bigwedge_{L \in \loopa{P}} \loopformula{L,P}$.
Let $M \subseteq \at{P}$ be a subset of atoms of $P$.  We use $\tau^M:
\at{P} \rightarrow \{0,1\}$ to denote the assignment corresponding to
$M$, i.e.  $\tau^M(v) = 1$ if $v \in M$ and $\tau^M(v) = 0$ otherwise,
for all $v \in \at{P}$.  Then $M$ is an answer set of $P$ if and only
if $\tau^M$ satisfies the propositional formula $\completion{P} \wedge
\loopformula{P}$~\cite{LZ2004}.

\subsection{Related Work}\label{sec:relatedwork}

The decision version of normal logic programs is NP-complete; therefore, the ASP counting for normal logic programs is \#P-complete~\cite{valiant1979}.  Given the \#P-completeness, a prominent line of work focused on ASP counting relies on translations from the ASP program to the CNF formula~\cite{LZ2004,Janhunen2004,Janhunen2006,JN2011}. Such translations often result in a large number of CNF clauses and thereby limit practical scalability for {\em non-tight} ASP programs. 
Eiter et al.~\shortcite{EHK2021} introduced T$_{\mathcal{P}}$-\textit{unfolding} to break cycles and produce a tight program. They proposed an ASP counter called aspmc, that performs a treewidth-aware Clark completion from a cycle-free program to the CNF formula. Jakl, Pichler, and Woltran~\shortcite{JPW09} extended the tree decomposition based approach for \#SAT due to Samer and Szeider~\shortcite{SS07} to Answer Set Programming and proposed a fixed-parameter tractable (FPT) algorithm for ASP counting. 
Fichte et al.~\shortcite{FHMW2017,FH2019} revisited the FPT algorithm due to Jakl et al. and  developed an exact model counter, called DynASP, that performs well on instances with low treewidth. 
Aziz et al.~\shortcite{ACMS2015} extended a propositional model counter to an answer set counter by integrating unfounded set detection. Kabir et al.~\shortcite{KESHFM2022} focused on lifting hashing-based techniques to ASP counting, resulting in an approximate counter, called ApproxASP, with $(\varepsilon,\delta)$-guarantees.

\section{An Alternative Definition of Answer Set}
\label{sec:algorithm}
Our algorithm for answer set counting crucially relies on an alternative way 
of defining the answer sets of a normal program
$P$.  We first introduce an operation called $\copyop{}$ that plays a
central role in this alternative definition.  Our $\copyop{}$ operation
is related to, but not the same as, a similar operation used
in~\asproblog. Specifically, founded variables (i.e. variables
appearing at the head of a rule) were the focus of the copy operation
used in~\asproblog.  In contrast, loop atoms in the program $P$
are the focus of the $\copyop{}$ operation in our approach.  We
elaborate more on this below.

\subsection{The Copy Operation}%

Given a normal program $P$, for every loop atom/variable $v$ in
$\loopatoms{P}$, let $\copyatom{v}$ be a fresh variable not present in
$\at{P}$. We refer to $\copyatom{v}$ as the \emph{copy variable of $v$}.
For $X \subseteq \loopatoms{P}$, 
we denote the set of copy variables corresponding to atoms in $X$ as $\copyatom{X}$.

The $\copyop{}$ operation, when applied to a normal program $P$,
returns a set of (implicitly conjoined) implications, defined as follows:
\begin{enumerate}
\item \label{l1:type1} (type 1) for every $v \in \loopatoms{P}$, the implication $\copyatom{v} \rightarrow v$ is in $\copyop{P}$.
\item \label{l1:type2} (type 2) for every rule $x \leftarrow a_1, \ldots a_k, b_1, \ldots b_m, \sim c_1, \ldots \sim c_n$ in $P$, where
  $x \in \loopatoms{P}$,
  $\{a_1, \ldots a_k\} \subseteq \loopatoms{P}$ and
  $\{b_1, \ldots b_m\} \cap \loopatoms{P} = \emptyset$,
  the implication $\copyatom{a_1} \wedge \ldots \copyatom{a_k} \wedge b_1 \wedge 
  \ldots b_m \wedge \neg{c_1} \wedge \ldots \neg{c_n} \rightarrow \copyatom{x}$ is in $\copyop{P}$.
\item No other implication is in $\copyop{P}$.
\end{enumerate}
Note that in implications of type~\ref{l1:type2}, copy variables are used
exclusively for positive loop atoms in the body of the rule and for the loop
atom in the head of the rule.  Specifically, if the head of a rule is
not a loop atom, we don't add any implication of type~\ref{l1:type2}
for that rule.  As
an extreme case, if $P$ is a tight program or $\loopatoms{P} =
\emptyset$, then $\copyop{P}$ is empty.

\paragraph{An Alternative Definition of Answer Set}
We now present a key observation that provides the basis for an alternative
definition of answer sets.  Akin to the existing definitions of answer
set~\cite{Janhunen2006,GLM2006,Lifschitz2010}, our definition seeks
justification for atoms within an answer set.  However, our definition
seeks to justify only loop atoms belonging to an answer set, while the
existing definitions, to the best of our knowledge, aim to justify
each atom in an answer set. The alternative definition derives from
the observation that under Clark's completion of a program, if the loop
atoms of an answer set are justified, then the remaining atoms of the
answer set are also justified. Thus, under Clark's completion, it
suffices to seek justifications for loop atoms. 
Unlike existing definitions of answer sets, our definition of
answer sets operates exclusively within the realm of Boolean formulas and 
employs unit propagation as a tool to decide whether an
atom is justified or not.

Recall from
Section~\ref{sec:preliminaries} the definition of $\phi|_{\tau}$,
i.e. unit propagation of an assignment $\tau$ on a CNF formula $\phi$.
Recall also that a CNF formula can be viewed as a set of clauses,
where each clause can be interpreted as an implication.  Therefore,
the set of implications $\copyop{P}$ can be thought of as representing
a CNF formula.  For an assignment $\tau: X \rightarrow \{0,1\}$ where
$X \subseteq \at{P}$, we use the notation $\copyop{P}|_{\tau}$ to
denote the (implicitly conjoined) set of implications that remain
after unit propagating $\tau$ on the CNF formula represented by
$\copyop{P}$.  Specifically, we say that $\copyop{P}|_{\tau} ~=
\emptyset$ if $\tau$ unit propagates to only unit clauses on copy
variables in the CNF formula represented by $\copyop{P}$.

\begin{theorem}
\label{theorem:sm_equal}
For a normal program $P$, let $X \subseteq \at{P}$ and let $\tau: X
\mapsto \{0,1\}$ be an assignment.  Let $M^\tau$ denote the set of
atoms of $P$ that are assigned $1$ by $\tau$.  Then $M^\tau \in
\stable{P}$ if and only if $\tau \models \completion{P}$ and
$\copyop{P}|_{\tau} = \emptyset$.
\end{theorem}
\begin{proof}

(i)~(proof of `if part') \textbf{Proof By Contradiction}. Assume that
  $\tau \models \completion{P}$ and $\copyop{P}|_{\tau} = \emptyset$,
  but $M^\tau \not\in \stable{P}$.  Since $M^\tau \not\in \stable{P}$ and
  $\tau \models \completion{P}$, it implies that $\tau \not\models
  \loopformula{P}$.  Thus, there is a loop $L$ in $P$ such that $\tau
  \not\models \loopformula{L,P}$.
Assume that $L$ is comprised of the set of loop atoms $\SB x_1, \ldots, x_k\SE$. Then $\tau \not\models x_1 \wedge 
\ldots \wedge x_k \rightarrow \bigvee_{r \in \external{L}} \body{r}$.
In other words, even if $\tau$ is augmented by setting $x_1 = \ldots =
x_k = 1$, the formula $\bigvee_{r \in \external{L}} \body{r}$
evaluates to $0$ under the augmented assignment. Now recall that
$\tau$ itself is an assignment to a subset of $\at{P}$, and it does
not assign any truth value to $\copyatom{x_1}, \ldots,
\copyatom{x_k}$. Therefore, there must be at least one type~\ref{l1:type2}
implication in $\copyop{P}|_{\tau}$, specifically one arising from a
rule $r \in \external{L}$, that does not unit propagate to a unit
clause or to $1$ under $\tau$.  This contradicts the premise that
$\copyop{P}|_{\tau}=\emptyset$.

(ii)~(proof of `only if part') \textbf{Proof By Contradiction}.
Suppose $M^\tau \in \stable{P}$.  We know that this implies $\tau
\models \completion{P} \wedge \loopformula{P}$.  We now show that in
this case, we must also have $\copyop{P}|_{\tau} = \emptyset$.
Suppose, if possible, $\copyop{P}|_{\tau} \neq \emptyset$. We ask if
an implication of type \ref{l1:type1}, say $\copyatom{v} \rightarrow
v$, can stay back in $\copyop{P}|_{\tau}$. If $v \in M^{\tau}$, then
$\tau(v) = 1$, and clearly the implication $\copyatom{v} \rightarrow
v$ doesn't stay back in $\copyop{P}|_{\tau}$.  If $v \not\in
M^{\tau}$, then $\tau(v) = 0$, and in this case $\tau$ unit propagates
to $\{\neg \copyatom{v}\}$, and hence the implication doesn't stay
back in $\copyop{P}|_{\tau}$ either.  Therefore, no implication of
type \ref{l1:type1} can stay back in $\copyop{P}|_{\tau}$.  Next, we
ask if any implication of type \ref{l1:type2} can stay back in
$\copyop{P}|_{\tau}$.  Suppose this is possible.  Note that for every
$v \in \at{P}$, either $v \in M^\tau$ or $v \not\in M^\tau$.
Therefore, $\tau(v)$ is either $0$ or $1$ for all $v \in \at{P}$.
Therefore, if $\copyop{P}|_{\tau} \neq \emptyset$, there must be some
$\copyatom{x_1} \in \var{\copyop{P}|_{\tau}}$ and there must a
(potentially simplified) implication $\copyatom{x_2} \wedge C_1
\rightarrow \copyatom{x_1}$ in $\copyop{P}|_{\tau}$, where $C_1$ is
either \true or a conjunction of copy variables.
The existence of copy variable $\copyatom{x_2}$ in
$\copyop{P}|_{\tau}$ implies the existence of another implication:
$\copyatom{x_3} \wedge C_2 \rightarrow \copyatom{x_2}$ in
$\copyop{P}|_{\tau}$.
Continuing this argument, we find that there are two cases to handle:~(i)~there are an unbounded
number of copy variables in $\copyop{P}|_{\tau}$, which contradicts
the fact that there can be at most
$\Card{\var{P}}$ copy variables.
(ii)~otherwise, there exists $i, j$ such that $\copyatom{x_i} = \copyatom{x_j}$ and $i < j$, which implies
that the set of variables $\SB x_i, \ldots x_{j-1} \SE$ constitutes an
\textit{unfounded set}.  However, this contradicts the fact that
$M^\tau \in \stable{P}$.
In either case, we reach a contradiction, thereby proving that
$\copyop{P}|_{\tau}$ is empty. This completes the proof.
\end{proof}

\begin{example}
\label{ex:least_model_1} 
Consider the normal program $P$ given by the rules $\{r_1=a \leftarrow
\sim b.~ \text{ } r_2=b \leftarrow \sim a.~ \text{ } r_3=c \leftarrow
a, b. \text{ }~ r_4=c \leftarrow d.~ \text{ } r_5=d \leftarrow a.~
\text{ } r_6=d \leftarrow b,c.~\text{ } r_7=e \leftarrow \sim a,\sim
b\}$.\\ This program has a single loop $L$ consisting of atoms $c$ and
$d$, i.e.  $\loopatoms{P} = \{c,d\}$.  Therefore, $\copyop{P}$ consists of the conjunction of 
implications: $\SB
\copyatom{c} \rightarrow c, \copyatom{d} \rightarrow d, a \wedge b
\rightarrow \copyatom{c}, \copyatom{d} \rightarrow \copyatom{c}, a
\rightarrow \copyatom{d}, b \wedge \copyatom{c} \rightarrow
\copyatom{d} \SE$.  Note that there are no variables $a', b', e'$ or
constraints involving them in $\copyop{P}$.  The followings are now
easily verified.
\begin{itemize}
    \item Consider $\tau_1$ that assigns $1$ to $b$ and $0$ to $a, c, d, e$. For the corresponding answer set $M^{\tau_1}$: $\SB b \SE$, $\copyop{P}|_{\tau_1} = \emptyset$
    \item Consider $\tau_2$ assigns $1$ to $a,c,d$ and $0$ to $b,e$.  For
      the corresponding answer set $M^{\tau_2}$: $\SB a,c,d \SE$, $\copyop{P}|_{\tau_2} = \emptyset$
    \item Consider $\tau_3$ that assigns $1$ to $b, c, d$ and $0$ to
      $a, e$. For the corresponding non-answer set $M^{\tau_3}$: $\SB
      b,c,d \SE$, $\copyop{P}|_{\tau_3} \neq \emptyset$
\end{itemize}
\end{example}

\section{Counting Answer Sets}
\label{sec:implemetation}
In this section, we first show how the alternative definition of answer sets
provides a new way to counting all answer sets of a given normal
program. Subsequently, we explore how off-the-shelf state-of-the-art
propositional model counters can be easily adapted to correctly count
answer sets by leveraging the alternative definition.

It is easy to see from~\Cref{theorem:sm_equal} that the count of
answer sets of a normal program $P$ can be obtained simply by counting
assignments $\tau \in 2^{\Card{\at{P}}}$ such that $\tau \models
\completion{P}$ and $\copyop{P}|_{\tau} = \emptyset$.
This motivates us to represent a normal program $P$ using a pair ($F$,
$G$), where $F = \completion{P}$ and $G = \copyop{P}$. Further, we discuss
below how key ideas in state-of-the-art propositional model counters
can be adapted to work with this pair representation of normal
programs to yield exact answer set counters.

\subsection{Decomposition}
Propositional model counters often \textit{decompose} the input CNF formula into \textit{disjoint subformulas} 
to boost up the counting efficiency~\cite{BP2000} -- for two formulas $\phi_1$ and $\phi_2$, if $\var{\phi_1} \cap \var{\phi_2} = \emptyset$, 
then $\phi_1$ and $\phi_2$ are \textit{decomposable}, i.e., we can count the number of models of $\phi_1$ and $\phi_2$ separately and multiply 
these two counts to get the number of models of $\phi_1 \wedge \phi_2$. 

Given a normal program, our proposed definition involves a pair of formulas: $F$ and $G$. Specifically, we define \emph{component decomposition} with respect to $(F, G)$ as follows:
\begin{definition}
\label{def:decomposition}
    $(F_1 \wedge F_2, G_1 \wedge G_2)$ is decomposable to $(F_1,G_1)$ and $(F_2,G_2)$ if and only if $(\var{F_1} \cup \var{G_1}) \cap (\var{F_2} \cup \var{G_2}) = \emptyset$.
\end{definition}
Finally, \Cref{lemma:establish_decomposition} offers evidence supporting the correctness of our proposed definition of decomposition in computing the number of answer sets.
\begin{proposition}
\label{lemma:establish_decomposition}
Let $(F_1 \wedge \ldots F_k,  G_1 \wedge \ldots G_k)$ is decomposed to $(F_1, G_1),  \ldots, (F_k, G_k)$ then 
$\countsm{F_1 \wedge \ldots F_k}{G_1 \wedge \ldots G_k} = \countsm{F_1}{G_1} \times \ldots \countsm{F_k}{G_k}$
\end{proposition}
\begin{proof}
By definition of decomposition, we know that
$\big(\var{F_i} \cup \var{G_i}\big) \cap \big(\var{F_j} \cup \var{G_j}\big)
= \emptyset$ for $1 \le i < j \le k$.  This, in turn, implies that
$\var{G_i} \cap \var{G_j} = \emptyset$ for $1 \le i < j \le k$.
Therefore, no variable (copy variable or otherwise) is common in $G_i$
and $G_j$, if $i \neq j$.  Hence, for every assignment
$\tau: \at{P} \rightarrow \{0,1\}$, unit propagation of $\tau$ on
$G_i$ and $G_j$ must happen completely independent of each other,
i.e. no unit literal obtained by unit propagation of $\tau$ on $G_i$
affects unit propagation of $\tau$ on $G_j$, and vice versa.  In other
words, $G_i|_{\tau} \wedge G_j|_{\tau} = (G_i \wedge G_j)|_{\tau}$.

Let $F = F_1 \wedge \ldots F_k$ and $G = G_1 \wedge \ldots G_k$.  In
the following, we use the notation $\tau$ to denote an assignment
$\at{P} \rightarrow \{0,1\}$, and $\tau_i$ to denote an assignment
$\at{P} \cap (\var{F_i} \cup \var{G_i}) \rightarrow \{0,1\}$, for
$1 \le i \le k$.  By virtue of the argument in the previous paragraph,
it is easy to see that the domains of $\tau_i$ and $\tau_j$ are
disjoint for $1 \le i < j \le k$.  We use the notation
$\tau_1 \cup \ldots \tau_k$ to denote the assignment
$\at{P} \rightarrow \{0,1\}$ defined as follows: if $v \in
\at{P} \cap \big(\var{F_i} \cup \var{G_i}\big)$, then
$(\tau_1 \cup \ldots \tau_k)(v) = \tau_i(v)$. The proof now consists
of showing the following two claims:
\begin{enumerate}
\item $\countsm{F_1}{G_1} \times \cdots \countsm{F_k}{G_k} \ge \countsm{F}{G}$.
\item $\countsm{F_1}{G_1} \times \cdots \countsm{F_k}{G_k} \le \countsm{F}{G}$.
\end{enumerate}
\noindent \textbf{Proof of part 1:}
Suppose $\tau \in \answer{F}{G}$.  By definition, $\tau \models F$ and
$G|_{\tau} = \emptyset$.  Since $F = F_1 \wedge \ldots F_k$, we know
that $\tau \models F_i$ for $1 \le i \le k$.  By the above definition
of $\tau_i$, it then follows that $\tau_i \models F_i$.  Similarly,
since unit propagation of $\tau$ on $G_i$ and $G_j$ happen
independently for all $i \neq j$, and since unit propagation of $\tau$
on $G = G_1 \wedge \ldots G_k$ gives $\emptyset$, we have
$G_i|_{\tau_i} = \emptyset$ as well.  It follows that
$\tau_i \in \answer{F_i}{G_i}$ for $1 \le i \le k$.  Therefore, every
$\tau \in \answer{F}{G}$ yields a sequence of
$\tau_i \in \answer{F_i}{G_i}$, for $1 \le i \le k$.  Since the
domains of all $\tau_i$'s are distinct, it follows that
$\countsm{F_1}{G_1} \times \cdots \countsm{F_k}{G_k} \ge \countsm{F}{G}$.

\noindent \textbf{Proof of part 2:}
Suppose $\tau_i \in \answer{F_i}{G_i}$ for $1 \le i \le k$.  By
definition, $\tau_i \models F_i$ and $G_i|_{\tau_i} = \emptyset$.
Since the domains of $\tau_i$ and $\tau_j$ are disjoint for all $1 \le
i < j \le k$, it follows that $(\tau_1 \cup \ldots \tau_k) \models
(F_1 \wedge \ldots F_k)$ and hence $\tau \models F$.  We have also
seen that $(G_1 \wedge \cdots G_k)|_{\tau} =
(G_1|_{\tau} \wedge \cdots G_k|_{\tau})$.  However, since $\var{G_i}$
is a subset of the domain of $\tau_i$, we have $(G_1 \wedge \cdots
G_k)|_{\tau} = (G_1|_{\tau_1} \wedge \cdots G_k|_{\tau_k})$.  Since
$G_i|_{\tau_i} = \emptyset$ for $1 \le i \le k$, it follows that
$(G_1 \wedge \cdots G_k)|_{\tau} = \emptyset$.  Therefore $G|_{\tau}
= \emptyset$.  Since $\tau \models F$ as well, we have
$\tau \in \answer{F}{G}$.  Therefore, every distinct sequence of
$\tau_i, 1 \le i \le k$ such that $\tau_i \in \answer{F_i}{G_i}$
yields a distinct $\tau \in \answer{F}{G}$.  It follows that
$\countsm{F_1}{G_1} \times \cdots \countsm{F_k}{G_k} \le \countsm{F}{G}$.

\noindent It follows from the above two claims that
$$\countsm{F_1}{G_1} \times \cdots \countsm{F_k}{G_k} = \countsm{F}{G}.$$

\end{proof}

	One of the drawbacks of the definition is its comparative weakness in relation to the conventional definition of decomposition. When dealing with a \textit{hard-to-decompose} program $(F,G)$, 
then the process of counting answer sets regresses to enumerating the answer sets of the program.

\subsection{Determinism}
Propositional model counters utilize \textit{determinism}~\cite{Darwiche2002}, which involves assigning one of the variables in a formula to either \false or \true.
The number of models of $\phi$ is then determined as the sum of the number of models in which a variable $x \in \var{\phi}$ is assigned to \false and \true.  A similar idea can be used for answer set counting using our pair representation as well.  To establish the correctness of the determinism employed in our approach, we first introduce two helper propositions: \Cref{lemma:proof_of_part_1} and \ref{lemma:proof_of_part_2}.

\begin{proposition}
	\label{lemma:proof_of_part_1}
For partial assignment $\tau$ and program $P$ represented as $(\completion{P}, \copyop{P})$,  
if $\completion{P}|_{\tau} = \emptyset$ and $\emptyset \subset \var{\copyop{P}|_\tau} \subseteq \copyvar{P}$, 
then $\exists L \in \loopa{P}$ s.t.  $L \subseteq M^\tau$ and  $\tau \not\models \loopformula{L,P}$.
\end{proposition}
\begin{proof}
	Since  $\emptyset \subset \var{\copyop{P}|_\tau} \subseteq \copyvar{P}$,  there exists a copy variable  $\copyatom{x_{i_1}} \in \var{\copyop{P}|_{\tau}}$ and an implication (simplified after unit propagation) of type~\ref{l1:type2} of the form $C_1 \rightarrow \copyatom{x_{i_1}}$ in $\copyop{P}|_{\tau}$, where $C_1$ is a non-empty conjunction of copy variables. Let $\copyatom{x_{i_2}} \in \var{C_1}$, then there must also exist another implication (simplified after unit propagation) $C_2 \rightarrow \copyatom{x_{i_2}}$ in $\copyop{P}|_{\tau}$, where $C_2$ is again a conjunction of copy variables. Accordingly, for $x_{i_3} \in \var{C_2} \setminus \var{C_1}$, we have another implication of the form $C_3 \rightarrow \copyatom{x_{i_3}}$ in $\copyop{P}|_{\tau}$.  
	Since the number of atoms is bounded, it must be the case that there exists $i_k$ such that there is an implication (simplified) of type~\ref{l1:type2} $C_k \rightarrow \copyatom{x_k}$ such that $C_k \setminus (C_1 \cup C_2 \ldots C_{k-1}) = \emptyset$. 
	
	Now, observation $C_k \setminus (C_1 \cup C_2 \ldots C_{k-1})
	= \emptyset$ implies existence of an atom set $L = \{x_{i_1},
	x_{i_2}, \ldots x_{i_{j}}\} \subseteq \{x_{i_1},
	x_{i_2}, \ldots x_{i_{k}}\}$ that forms a loop in
	$\dependency{P}$.  Given that
	$\var{\copyop{P}|_\tau} \subseteq \copyvar{P}$, we also know
	that $\tau$ assigns a value to every
	$x \in \var{\copyop{P}} \cap \at{P}$.  Furthermore, each of
	the atoms $x_{i_1}, \ldots x_{i_k}$ must have been assigned
	$1$ by $\tau$.  Otherwise, if any $x_{i_l}$ was assigned $0$
	by $\tau$, then $\tau$ would have unit propagated on
	$\copyop{P}|_\tau$ to $\neg x_{i_l}'$, which contradicts the
	observation that the copy variables $x_{i_1}', \ldots
	x_{i_k}'$ stayed backed in antecedents of implications of
	type \ref{l1:type2} in $\copyop{P}|_{\tau}$.  It follows
	that atoms in loop $L$ form a subset of atoms assigned
	$1$ by $\tau$.

	We have shown above that $\{x_{i_1}, x_{i_2}, \ldots
	x_{i_{j}}\}$ constitutes a loop in the positive dependency
	graph.  We now show by contradiction that
	$\tau \not\models \bigvee_{r \in \external{L}} \body{r}$. Indeed,
	if $\tau \models \bigvee_{r \in \external{L}} \body{r}$, let
	$x_{i_t}$ be $\head{r}$ for a rule $r$ such that
	$\tau \models \body{r}$. In this case, $\tau$ must have unit
	propagated to $\{x_{i_t}'\}$ in $\copyop{P}|_{\tau}$.  This
	contradicts the fact that the copy variables $x_{i_1}', \ldots
	x_{i_k}'$ stayed backed in antecedents of implications of
	type \ref{l1:type2} in $\copyop{P}|_{\tau}$.

	Therefore $\tau \models x_{i_1} \wedge \ldots x_{i_j}$ but
	$\tau \not\models \bigvee_{r \in \external{L}} \body{r}$.  This
	shows that $\tau \not\models \loopformula{L,P}$.
\end{proof}
\begin{proposition}
	\label{lemma:proof_of_part_2}
	For partial assignment $\tau$  and program $P$ represented as 
	$(\completion{P}, \copyop{P})$,  suppose $\tau \not \models \loopformula{L,P}$,
	 where $L = \{x_1, \ldots, x_k\}$.  Then there exists $\tau^+$ such that $\{\copyatom{x_1}, \ldots, \copyatom{x_k}\} 
	\subseteq \var{\copyop{P}|_{\tau^+}}$, $\completion{P}|_{\tau^+} = \emptyset$ and $\tau \subseteq \tau^+$. 
\end{proposition}
\begin{proof}

	As $\tau \not \models \loopformula{L,P}$,  we have $\forall x_i \in L, \tau(x_i) = 1$ and $\forall r \in \external{L}, \tau \not\models \body{r}$. Let us denote by $r'$ an implication  of type~\ref{l1:type2} corresponding to a rule $r \in \external{L}$.  Then we have $r'|_{\tau} \neq \emptyset$; moreover, if $\head{r} = x_i$, then $\copyatom{x_i} \in \var{r|_{\tau}}$. Since the above observation holds for all $r \in \external{L}$ and for $x_i \in L$, therefore,  $\{\copyatom{x_1}, \ldots, \copyatom{x_k}\} 
	\subseteq \var{\copyop{P}|_{\tau}}$. Observe that for every extension $\tau'$ of $\tau$ that does not assign values to variables in $\{\copyatom{x_1}, \ldots, \copyatom{x_k}\}$, it must be the case that  $\{\copyatom{x_1}, \ldots, \copyatom{x_k}\} 
	\subseteq \var{\copyop{P}|_{\tau'}}$. Furthermore, since the set of variables in $\completion{P}$ does not contain a variable from the set  $\{\copyatom{x_1}, \ldots, \copyatom{x_k}\}$, therefore, there exists an extension, $\tau^{+}$, of $\tau$ such that  $\completion{P}|_{\tau^+} = \emptyset$ and $\{\copyatom{x_1}, \ldots, \copyatom{x_k}\} 
	\subseteq \var{\copyop{P}|_{\tau'}}$.
\end{proof}
\noindent We are now ready to state and prove the correctness of determinism employed in our ASP counter:

\begin{proposition}
  \label{lemma:determinism_equation}
  Let program $P$ be represented as $(F, G)$. Then
	\begin{align}
		\countsm{F}{G} &= \countsm{F|_{\neg{x}}}{G|_{\neg{x}}} + \countsm{F|_{x}}{G|_{x}},\nonumber \\ &~~ \text{for all } x \in \at{P} \label{eq:determinism} \\
		\countsm{\bot}{G} &= 0 \label{eq:base_case_1}\\ 
		\countsm{\emptyset}{G} &= 
		\begin{cases}
			1 & \text{if $G = \emptyset$ } \\
			0 & \text{if $\var{G} \subseteq \copyvar{P}$ }
		\end{cases} \label{eq:base_case_2}
	\end{align}
\end{proposition} 
\noindent Note that if $\completion{P} = \emptyset$ then either $G = \emptyset$ or $\emptyset \subset \var{G} \subseteq \copyvar{P}$.
\begin{proof}
	The proof comprises the following three parts:
	
	\Cref{eq:determinism} applies determinism by partitioning all answer sets of $(F,G)$ into two parts --
	the answer sets where $x$ is $0$ and $1$, respectively. Observe that 
	performing unit propagation on $(F,G)$ is valid since $\tau \in \stable{F|_{\sigma},G|_{\sigma}}$ if and only if $\sigma \cup \tau \in \stable{F,G}$, where $\sigma \in 2^{\Card{X}}$, $\tau \in 2^{\Card{\at{P}\setminus X}}$, where $X \subseteq \at{P}$.
	
	The proof of the first base case \cref{eq:base_case_1} is trivial. Each answer set of $P$ conforms to the completion of the program $\completion{P}$, where, according to the alternative definition of answer sets, $F = \completion{P}$.
	
	We utilize the helper propositions proved earlier to demonstrate the correctness of the second base case, as outlined in \cref{eq:base_case_2}, 
	which appropriately selects answer sets from the models of completion.
	First, we show that if there is a copy variable in $\copyop{P}|_{\tau}$, where $\completion{P}|_{\tau} = \emptyset$, then one of the loop 
	formulas of the program is not satisfied by $\tau$. The claim is proved in~\Cref{lemma:proof_of_part_1}. Thus, $\tau$
	cannot be extended to an answer set. Second, we demonstrate that if there is an unsatisfied loop formula under a partial
	assignment $\tau_1$, then there exists $\tau_{1}^+$ such that some copy variables are not propagated in $\copyop{P}|_{\tau_{1}^+}$, where $\completion{P}|_{\tau_{1}^+} = \emptyset$ and $\tau_1 \subseteq \tau_{1}^+$.
	The claim is established in~\Cref{lemma:proof_of_part_2}.  Thus, through the method of contradiction, we can infer that, for an assignment $\tau$, if $\copyop{P}|_{\tau} = \emptyset$, then 
	$\tau$ can be extended to an answer set.
	
\end{proof}

\subsection{Conjoin $F$ and $G$}
Until now, we have represented a program $P$ as a pair of formulas $F$ and $G$. However, in this subsection, 
we illustrate that rather than considering the pair, we can regard their conjunction $F \wedge G$, 
and all the subroutines of model counting algorithms work correctly. First,  in~\Cref{lemma:work_with_conjunction}, we demonstrate 
that $F \wedge G$ uniquely defines a program $(F, G)$ under arbitrary partial assignments.

\begin{proposition}
\label{lemma:work_with_conjunction}
For two assignments $\tau_1$ and $\tau_2$, and given a normal program,  $F|_{\tau_1} \wedge G |_{\tau_1} = F|_{\tau_2} \wedge G |_{\tau_2}$ if and only if $F|_{\tau_1} = F|_{\tau_2}$ and $G|_{\tau_1} = G|_{\tau_2}$
\end{proposition}
\begin{proof}
	
	(i)~(proof of `if part') The proof is trivial.
	
	(ii)~(proof of `only if part') \textbf{Proof By Contradiction}. Assume that there is a clause 
	$c \in F|_{\tau_1}$ and $c \not\in F|_{\tau_2}$. As $F|_{\tau_1} \wedge G |_{\tau_1} = F|_{\tau_2} \wedge G |_{\tau_2}$ 
	clause $c \in G|_{\tau_2}$. As $c \in F|_{\tau_1}$, $c$ has no copy variable. Assume that 
	clause $c$ is derived from the unit propagation of $\copyop{r}$, i.e., $c = \copyop{r}|_{\tau_2} 
	= \copyatom{a_1} \wedge \ldots \wedge \copyatom{a_k} \wedge b_1 \wedge \ldots \wedge b_m \wedge \neg{c_1} 
	\wedge \ldots \wedge \neg{c_n} \rightarrow \copyatom{x}|_{\tau_2}$, where $\forall i, \copyatom{a_i}$ 
	propagates to $1$ and $\copyatom{x}$ propagates to $0$, which follows that under assignment $\tau_2$,
	the atom $x$ is assigned to $0$ and $\forall i, a_i$ is assigned to $1$. 
	The rule $r$ also belongs to $\completion{P}$ and 
	both $F|_{\tau_1}$ and $F|_{\tau_2}$ are derived from $\completion{P}$. 
	Thus, under assignment $\tau_2$, if $x$ is assigned to $0$ and each of the $a_i$'s is assigned to $1$, then 
	the clause $c \in  F|_{\tau_2}$, which must be derived from rule $r$, so contradiction.
\end{proof}

As a result, it is possible to perform unit propagation on $F \wedge G$ instead of performing unit propagation on $F$ and $G$ separately.
Although both formulas $F$ and $G$ are necessary to check the base cases, we can still check base cases by considering the conjunction $F \wedge G$. 
Checking the first base case (\cref{eq:base_case_1}) is trivial because if an assignment $\tau$ \textit{conflicts} on $F$, then $\tau$ conflicts on $F \wedge G$ as well.
Additionally, calculating $\var{F \wedge G}$ suffices to check the second base case (\cref{eq:base_case_2}). 
The component decomposition part also works with their conjunction 
because the component decomposition condition $(\var{F_1} \cup \var{G_1}) \cap (\var{F_2} \cup 
\var{G_2}) = \emptyset$ is equivalent to $(\var{F_1 \wedge G_1}) \cap (\var{F_2 \wedge G_2}) = 
\emptyset$. Moreover, as we restrict our decision to $\at{P}$, the conjunction $F \wedge G$ does not introduce new conflicts --- if a partial 
assignment $\tau$ conflicts on $F \wedge G$, then $\tau$ conflicts on $F$.
To summarize, the model counting algorithm correctly computes the answer set count, even when processing the formula
$F \wedge G$ instead of processing the two formulas $F$ and $G$ separately.

\subsection{{\toolname}: Putting It All Together}
In this subsection, we aim to extend a propositional model counter to an exact answer set counter by integrating the alternative answer set definition, component decomposition 
(\Cref{lemma:establish_decomposition}), and determinism (\Cref{eq:determinism}). 

\begin{algorithm}[h]
	\caption{$\sharpasp(P)$}
	\label{alg:initial_computation}

	\begin{algorithmic}[1]
	\Function{$\counter{\phi}{CV}$}{} \Comment{modified CNF counter}
	\If {$\phi = \emptyset$} 
		\Return{$1$}	
	\ElsIf  {$\var{\phi} \subseteq CV$}
		\Return{$0$}	\label{line:base_case}
	\ElsIf {$\emptyset \in \phi$}
		\Return{$0$}	
	\EndIf
	\State $\mathsf{v} \gets \decide{\phi}$  \label{line:branching}
	\For {$\ell \gets \{\mathsf{v},  \neg{\mathsf{v}}\}$}
		\State $\cnt[\ell] \gets 1$
		\State $\mathsf{comps} \gets \decompose{\phi|_{\ell}}$ \label{line:decomposition}
		\ForEach {$\mathsf{c} \in \mathsf{comps}$} 
		\If {$\mathsf{c} \in \cache$} \label{line:caching}
		\State $\cnt[\ell] \gets \cnt[\ell] \times \cache[\mathsf{c}]$
		\Else
		\State $\cnt[\ell] \gets \cnt[\ell] \times \counter{\mathsf{c}}{CV}$
		\EndIf
		\If {$\cnt[\ell] = 0$}
		\State \textbf{break}
		\EndIf
		\EndFor
	\EndFor
	\State $\cache[\phi] \gets \cnt[\mathsf{v}] + \cnt[\neg{\mathsf{v}}]$
	\State \Return{$\cache[\phi]$}
	\EndFunction
	\State $F \gets \completion{P},  G \gets \copyop{P}$\label{line:compcopy} \Comment{Algorithm starts here}
	\State \Return{$\counter{F \wedge G}{\copyvar{P}}$}  \label{line:counterinvoke}
	\end{algorithmic}
\end{algorithm}

The pseudocode for {\toolname} is presented in~\Cref{alg:initial_computation}.
Given a non-tight program $P$, {\toolname} initially computes $\completion{P}$ 
and $\copyop{P}$ (Line~\ref{line:compcopy} of \Cref{alg:initial_computation}) and then calls the adapted propositional model counter
$\mathsf{Counter}$, with $\completion{P} \wedge \copyop{P}$ as the input formula, 
and $\copyvar{P}$ as the set of copy variables (Line \ref{line:counterinvoke} of \Cref{alg:initial_computation}). The model counting algorithm utilizes 
$\copyvar{P}$ to check the base cases (Equation (\ref{eq:base_case_1}) and (\ref{eq:base_case_2})) of the~\Cref{eq:determinism}.

The $\mathsf{Counter}$ differs from the existing propositional model counters mainly in two ways. Firstly,  
following~\cref{eq:base_case_2}, the $\mathsf{Counter}$ returns $0$ if it encounters a component  
consisting solely of copy variables (Line \ref{line:base_case} of \Cref{alg:initial_computation}).
Secondly, during \textit{variable branching}, $\mathsf{Counter}$ selects 
variables from $\var{\completion{P}}$ (Line~\ref{line:branching} of 
\Cref{alg:initial_computation}). Apart from that, the subroutines of unit propagation, 
component decomposition (Line~\ref{line:decomposition} of~\Cref{alg:initial_computation}), and caching\footnote{Model counter stores the count of previously 
solved subformulas by a caching mechanism to avoid recounting.} 
(Line~\ref{line:caching} of~\Cref{alg:initial_computation}) within $\mathsf{Counter}$ and a propositional model counter remain unchanged.

While {\toolname} uses copy variables and copy operations similar to~\asproblog, 
there are notable distinctions between the two approaches.
Firstly, {\toolname} aims to justify only loop atoms, whereas the~\asproblog~algorithm aims
to justify all founded variables.
Our empirical findings underscore that loop atoms constitute a relatively small subset of the founded variables. 
Consequently, the copy operation of~\asproblog~introduces more copy variables and logical implications involving copy variables compared to ours. 
Secondly, the unit propagation techniques employed in~\asproblog~differ from those used in {\toolname}. Specifically,~\asproblog~
performs unit propagation by propagating only the justified literals from a program while leaving the unjustified literals 
in the residual program. In contrast, {\toolname} adheres to the conventional unit propagation technique and employs copy variables to determine whether all atoms are justified.

\section{Experimental Evaluation}
\label{sec:experiment}

We developed a prototype\footnote{Available at \url{https://github.com/meelgroup/sharpASP}} of {\toolname} on top of the existing state-of-the-art model counters ({\ganak}, {\dfour}, and {\sharpsattd})~\cite{KJ2021,SRSM2019,LM2017}. We modified {\sharpsattd} by disabling all the preprocessing techniques, as they would no longer preserve answer sets.  
We use notations {\toolname}({\shortsharpsattd}), {\toolname}({\shortganak}), and {\toolname}(D) to represent {\toolname} 
with underlying propositional model counters \sharpsattd, \ganak, and \dfour, respectively. 

We compared the performance of {\toolname} with that of the prior state-of-the-art exact ASP counters:  
\clingo\footnote{\clingo counts answer sets via enumeration.}~\cite{GKNS2007}, \asproblog~\cite{ACMS2015}, and DynASP~\cite{FHMW2017}. In addition, we utilized two translations from ASP to SAT:~(i)~lp2sat~\cite{Fages94,JN2011,Bomanson2017}
~(ii)~aspmc~\cite{EHK2021}, followed by invoking off-the-shelf propositional model counters. 
We use notations lp2sat+X and aspmc+X to denote lp2sat and aspmc followed by propositional model counter X, respectively.

Our benchmark suite consists of non-tight programs from the domains of the Hamiltonian cycle and graph reachability problems~\cite{KESHFM2022,ACMS2015}. 
We also considered the benchmark set from~\cite{EHK2021} (designated as aspben). 
We gathered a total of $1470$ graph instances from the benchmark set of~\cite{KESHFM2022,EHK2021}.  
The \textit{choice}/\textit{random} variables in the hamiltonian cycle and aspmc benchmark pertain to graph edges. 
While the choice variables are associated with graph nodes for the graph reachability problem.

All experiments were carried out on a high-performance computer cluster, 
where each node consists of AMD EPYC $7713$ CPUs running with $128$ real cores. 
The runtime and memory limit were set to $5000$ seconds and $8$GB, respectively.

\subsection{Runtime Performance Comparison}

    \begin{table*}[t]
      \centering
      \begin{tabular}{m{3em} m{3em} m{3em} m{3em} m{3em} m{2em} m{3em} m{3em} m{2em} m{3em} m{3em} m{2em} m{3em}}
      \toprule
      &  &  &  & \multicolumn{3}{|c|}{aspmc} & \multicolumn{3}{|c|}{lp2sat} & \multicolumn{3}{|c|}{{\toolname}} \\
      & clingo & \asproblogshort & DynASP & D & \shortganak & \shortsharpsattd & D & \shortganak & \shortsharpsattd & D & \shortganak & \shortsharpsattd \\
      Hamil. (405)  & 230 & 0 & 0 & 173 & 197 & 167 & 135 & 164 & 112 & 238 & 261 & 300\\
      Reach. (600) & 318 & 149 & 2 & 187 & 288 & 421 & 317 & 471 & 167 & 293 & 458 & 463\\
      aspben (465) & 321 & 39 & 208 & 278 & 285 & 252 & 278 & 193 & 193 & 282 & 273 & 260\\
      \midrule \midrule 
      Total (1470) & 869 (4285) & 188 (8722) & 210 (8571) & 638 (5829) & 770 (5015) & 840 (4572) & 730 (5282) & 668 (5734) & 776 (5082) & 813 (4514) & 992 (3473) & 1023 (3372)\\
      \bottomrule
      \end{tabular}
      \caption{The performance comparison of {\toolname} vis-a-vis other ASP counters on different problems in terms of number of solved instances and PAR$2$ scores.
      }
      \label{table:comparison_counter_problem_wise}
      \end{table*}
    \begin{table}[t]
      \centering
      \begin{tabular}{m{3em} m{2em} m{3em} m{3em} m{3em} m{4em}}
      \toprule
      & & \multicolumn{4}{|c|}{clingo ($\leq 10^5$) + }\\ 
      & \shortstack{clingo} & \shortstack{\asproblogshort} & \shortstack{aspmc\\+STD} & \shortstack{lp2sat\\+STD} & \shortstack{{\toolname}\\(STD)}\\
      \midrule
      Hamil. (405) & 230 & 123 & 167 & 128 & 302\\
      Reach. (600) & 318 & 152 & 418 & 470 & 460\\
       aspben (465) & 321 & 278 & 284 & 297 & 300\\
      \midrule \midrule 
      Total (1470) & 869 (4285) & 553 (6239) & 869 (4310) & 895 (4205) & \textbf{1062} (3082)\\
      \bottomrule
      \end{tabular}
      \caption{The performance comparison of hybrid counters in terms of number of solved instances and PAR$2$ scores.
      The hybrid counters correspond to last $4$ columns that employ clingo enumeration followed by ASP counters.
      The clingo ($2$nd column) refers to clingo enumeration for $5000$ seconds. 
      }
      \label{table:hybrid_counter_performance}
    \end{table}
      \begin{figure*}
        \centering
        \begin{subfigure}[t]{0.33\textwidth}
          \centering
            \includegraphics[width=0.98\linewidth]{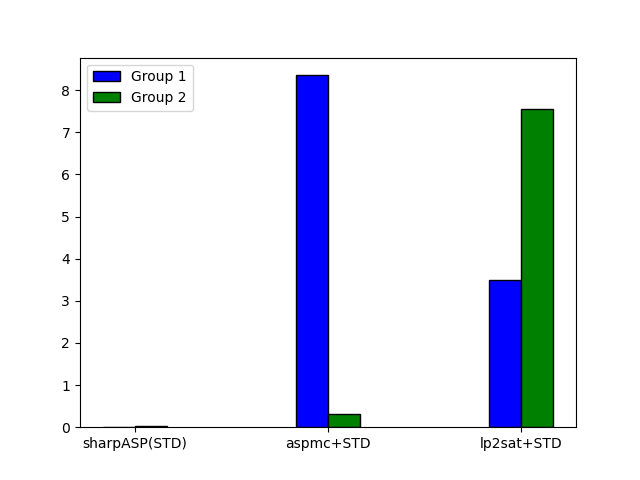}
          \caption{Time spent in BCP (seconds)
          }
        \label{fig:bcp_time}
        \end{subfigure}
        \begin{subfigure}[t]{0.33\textwidth}
          \centering
            \includegraphics[width=0.98\linewidth]{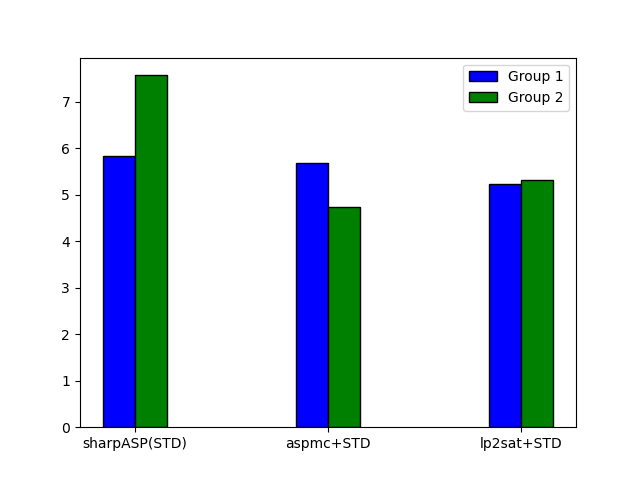}
          \caption{Number of decisions ($10$-base log).
          }
        \label{fig:decision}
        \end{subfigure}
        \begin{subfigure}[t]{0.33\textwidth}
          \centering
            \includegraphics[width=0.98\linewidth]{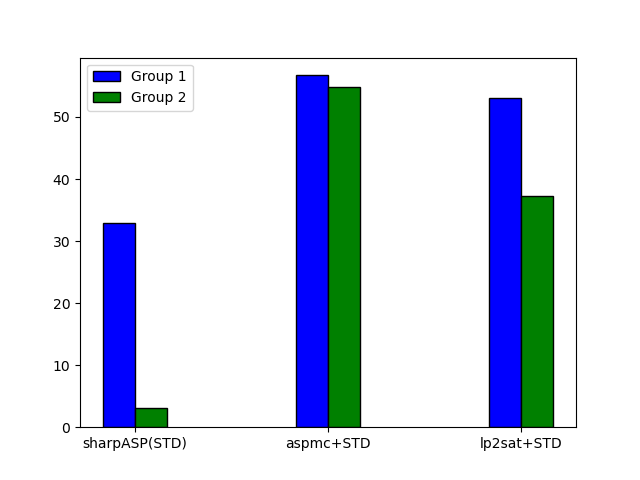}
          \caption{Cache hit (percentage).
          }
        \label{fig:cache_hit}
        \end{subfigure}
        \caption{The ablation study of {\toolname}(\shortsharpsattd), \asptosharpsat, and \aspmcsharpsat on Group $1$ and Group $2$ benchmarks.}
        \label{fig:ablation_study_across}
      \end{figure*}

The performance of our considered counters varies across different computational problems. Our evaluation of their performance, considering both total solved instances and PAR$2$ scores\footnote{
PAR$2$ is a penalized average runtime that 
penalizes two times the timeout for each unsolved benchmarks.},
for each computational problem is detailed in~\Cref{table:comparison_counter_problem_wise}.
The table demonstrates that {\toolname} either outperforms or achieves performance on par with existing ASP counters, particularly for the Hamiltonian cycle and graph reachability problems.
However, on aspben, the clingo enumeration outperforms other answer set counters.

We observed that clingo demonstrates superior performance, particularly on instances with a limited number of answer sets. 
Since this observation applies to all non-enumeration based counters in our repertoire, we devised a hybrid counter that combines the strengths of enumeration based counting with that of translation and propositional SAT based counting. 
Based on data collected from runs of clingo, there is a shift in the runtime performance of clingo when the count of 
answer sets exceeds $10^5$ (within our benchmarks). To ensure that our experiments can be replicated on different 
platforms, we chose to use an answer set count-based threshold instead of a time-based threshold.  
Hence, our hybrid counter is structured as follows: it initiates enumeration with a maximum of 
$10^5$ answer sets. In cases where not all answer sets are enumerated, 
the hybrid counter then employs an ASP counter with a time limit of $5000-t$ seconds, where $t$ is the time spent in clingo. 
The performance of the hybrid counters is tabulated in~\Cref{table:hybrid_counter_performance}, 
demonstrating that the hybrid counter based on {\sharpasp} clearly outperforms competitors by a handsome margin.

\subsection{Ablation Study}

We now delve into the internals, and to this end,  we  
form two groups of benchmarks -- \textbf{Group $1$}: instances where {\toolname}(\shortsharpsattd) runs faster than \asptosharpsat and \aspmcsharpsat, which highlights 
the scenarios where the {\toolname}(\shortsharpsattd) algorithm is more efficient than \asptosharpsat and \aspmcsharpsat; and \textbf{Group $2$}: instances where 
\asptosharpsat and \aspmcsharpsat run faster than {\toolname}(\shortsharpsattd), which shows the opposite scenario of Group $1$. 
Each group consists of 
$10$ instances that had more than $10^5$ answer sets, and therefore clingo could not enumerate all answer sets. By running the instances on all versions of \sharpsattd, we record the time spent on the procedure 
\textit{binary constraint propagation} (BCP), number of decisions, and \textit{cache hit rate} for each counter. 
Taking each group's average of each quantity provides a clear and concise
way to see how {\toolname} compares with others on average
across all benchmarks.  The statistical findings across all counters are visually summarized in~\Cref{fig:ablation_study_across}.

The strength of {\toolname} lies in its ability to minimize the time spent on binary constraint propagation (BCP) compared to other counters. 
The significantly large formula size increases the overhead for BCP in the case of {\asptosharpsat} and {\aspmcsharpsat}. 
However, we also observe that {\toolname} suffers from high overhead in the branching phase 
and high {\em cache misses} on Group $2$ instances. To find out the reason for a higher number of decisions, 
we analyze the decomposibility of Group $1$ and Group $2$ instances.

Our investigation has shown that, on all variants of {\sharpsattd}, most instances of Group $1$ start decomposing at nearly the same decision levels.
Thus, {\toolname}(\shortsharpsattd) outperforms on Group $1$ instances 
due to spending less time on BCP. We observed that several instances of Group $1$ took comparatively more decisions to make to 
count the number of answer sets on {\toolname}(\shortsharpsattd). One possible explanation is that \aspmcsharpsat and \asptosharpsat assign auxiliary variables, which have higher  
{\em activity scores} compared to original ASP program variables. 
Assigning auxiliary variables facilitates \asptosharpsat and \aspmcsharpsat by assigning fewer variables. 
However, {\toolname}(\shortsharpsattd) outperforms others due to structural simplicity and 
low-cost BCP. 

Our investigation has also revealed that Group $2$ instances are hard-to-decompose on {\toolname}(\shortsharpsattd) compared to other counters -- necessitating more variable assignments to break down an instance into disjoint components.
Since {\toolname}(\shortsharpsattd) assigns the original set of variables; it necessitates a larger number of decisions to count answer sets on hard-to-decompose instances compared to aspmc and lp2sat based counters. 
Moreover, the structure of hard-to-decompose instances also worsens the cache performance of {\toolname}.
However, \asptosharpsat and \aspmcsharpsat effectively decompose the input formula by initially assigning auxiliary variables.

In light of these findings, it is evident that the performance of {\toolname} is critically reliant on the decomposability of input instances and the variable branching 
heuristic employed. Notably, {\toolname} demonstrates superior performance when applied to {\em structurally simpler} input instances. If a variable branching heuristic effectively decomposes the input formula by assigning variables within the ASP programs, {\toolname} outperforms alternative ASP 
counters. Conversely, when the input formula's decomposability is hindered, alternative approaches involving the introduction of auxiliary variables prove to be more advantageous.

\section{Conclusion}\label{sec:conclusion}
Our approach, called {\toolname}, lifts the component caching-based propositional model counting to ASP counting without incurring a blowup in the size of the resulting formula. The proposed approach utilizes an alternative definition for answer sets, which enables the natural lifting of decomposability and determinism.  
Empirical evaluations show that {\toolname} and its corresponding hybrid solver can handle a greater number of instances compared to other techniques. 
As a future avenue of research, we plan to investigate extensions of our approach in the context of disjunctive programs. 

\paragraph{Acknowledgement}
This work was supported in part by National Research Foundation Singapore under its NRF Fellowship 
Programme [NRF-NRFFAI1-2019-0004], Ministry of Education Singapore Tier $2$ grant MOE-T2EP20121-0011, 
and Ministry of Education Singapore Tier $1$ Grant [R-252-000-B59-114]. 
The computational work for this article was performed on resources of the National Supercomputing 
Centre, Singapore (\url{https://www.nscc.sg}).
\bibliography{aaai24}

\begin{thebibliography}{36}
\providecommand{\natexlab}[1]{#1}

\bibitem[{Aziz et~al.(2015)Aziz, Chu, Muise, and Stuckey}]{ACMS2015}
Aziz, R.~A.; Chu, G.; Muise, C.; and Stuckey, P.~J. 2015.
\newblock Stable model counting and its application in probabilistic logic
  programming.
\newblock In \emph{AAAI}.

\bibitem[{Baluta et~al.(2019)Baluta, Shen, Shinde, Meel, and
  Saxena}]{BSSMS2019}
Baluta, T.; Shen, S.; Shinde, S.; Meel, K.~S.; and Saxena, P. 2019.
\newblock Quantitative verification of neural networks and its security
  applications.
\newblock In \emph{Proceedings of the 2019 ACM SIGSAC Conference on Computer
  and Communications Security}, 1249--1264.

\bibitem[{Bayardo~Jr and Pehoushek(2000)}]{BP2000}
Bayardo~Jr, R.~J.; and Pehoushek, J.~D. 2000.
\newblock Counting models using connected components.
\newblock In \emph{AAAI/IAAI}, 157--162.

\bibitem[{Biondi et~al.(2018)Biondi, Enescu, Heuser, Legay, Meel, and
  Quilbeuf}]{BEHLMQ2018}
Biondi, F.; Enescu, M.~A.; Heuser, A.; Legay, A.; Meel, K.~S.; and Quilbeuf, J.
  2018.
\newblock Scalable approximation of quantitative information flow in programs.
\newblock In \emph{VMCAI}, 71--93. Springer.

\bibitem[{Bomanson(2017)}]{Bomanson2017}
Bomanson, J. 2017.
\newblock lp2normal---A Normalization Tool for Extended Logic Programs.
\newblock In \emph{LPNMR}, 222--228.

\bibitem[{Brik and Remmel(2015)}]{BR2015}
Brik, A.; and Remmel, J. 2015.
\newblock Diagnosing automatic whitelisting for dynamic remarketing ads using
  hybrid ASP.
\newblock In \emph{LPNMR}, 173--185. Springer.

\bibitem[{Brooks et~al.(2007)Brooks, Erdem, Erdo{\u{g}}an, Minett, and
  Ringe}]{BEEMR2007}
Brooks, D.~R.; Erdem, E.; Erdo{\u{g}}an, S.~T.; Minett, J.~W.; and Ringe, D.
  2007.
\newblock Inferring phylogenetic trees using answer set programming.
\newblock \emph{Journal of Automated Reasoning}, 39(4): 471.

\bibitem[{Clark(1978)}]{clark1978}
Clark, K.~L. 1978.
\newblock Negation as failure.
\newblock In \emph{Logic and data bases}, 293--322. Springer.

\bibitem[{Darwiche(2002)}]{Darwiche2002}
Darwiche, A. 2002.
\newblock A compiler for deterministic, decomposable negation normal form.
\newblock In \emph{AAAI/IAAI}, 627--634.

\bibitem[{Dodaro and Maratea(2017)}]{DM2017}
Dodaro, C.; and Maratea, M. 2017.
\newblock Nurse scheduling via answer set programming.
\newblock In \emph{LPNMR}, 301--307. Springer.

\bibitem[{Eiter, Hecher, and Kiesel(2021)}]{EHK2021}
Eiter, T.; Hecher, M.; and Kiesel, R. 2021.
\newblock Treewidth-aware cycle breaking for algebraic answer set counting.
\newblock In \emph{KR}, volume~18, 269--279.

\bibitem[{Fages(1994)}]{Fages94}
Fages, F. 1994.
\newblock Consistency of {C}lark's completion and existence of stable models.
\newblock \emph{Journal of Methods of logic in computer science}, 1(1): 51--60.

\bibitem[{Fichte and Hecher(2019)}]{FH2019}
Fichte, J.~K.; and Hecher, M. 2019.
\newblock Treewidth and counting projected answer sets.
\newblock In \emph{LPNMR}, 105--119. Springer.

\bibitem[{Fichte et~al.(2017)Fichte, Hecher, Morak, and Woltran}]{FHMW2017}
Fichte, J.~K.; Hecher, M.; Morak, M.; and Woltran, S. 2017.
\newblock Answer Set Solving with Bounded Treewidth Revisited.
\newblock In \emph{{LPNMR}}, 132--145.

\bibitem[{Gebser et~al.(2007)Gebser, Kaufmann, Neumann, and Schaub}]{GKNS2007}
Gebser, M.; Kaufmann, B.; Neumann, A.; and Schaub, T. 2007.
\newblock Conflict-driven Answer Set Enumeration.
\newblock In \emph{LPNMR}, volume 4483, 136--148.

\bibitem[{Gelfond and Lifschitz(1988)}]{GL1988}
Gelfond, M.; and Lifschitz, V. 1988.
\newblock The stable model semantics for logic programming.
\newblock In \emph{ICLP/SLP}, volume~88, 1070--1080.

\bibitem[{Giunchiglia, Lierler, and Maratea(2006)}]{GLM2006}
Giunchiglia, E.; Lierler, Y.; and Maratea, M. 2006.
\newblock Answer set programming based on propositional satisfiability.
\newblock \emph{Journal of Automated reasoning}, 36(4): 345--377.

\bibitem[{Jakl, Pichler, and Woltran(2009)}]{JPW09}
Jakl, M.; Pichler, R.; and Woltran, S. 2009.
\newblock Answer-Set Programming with Bounded Treewidth.
\newblock In \emph{IJCAI}, volume~9, 816--822.

\bibitem[{Janhunen(2004)}]{Janhunen2004}
Janhunen, T. 2004.
\newblock Representing normal programs with clauses.
\newblock In \emph{ECAI}, volume~16, 358.

\bibitem[{Janhunen(2006)}]{Janhunen2006}
Janhunen, T. 2006.
\newblock Some (in) translatability results for normal logic programs and
  propositional theories.
\newblock \emph{Journal of Applied Non-Classical Logics}, 16(1-2): 35--86.

\bibitem[{Janhunen and Niemel{\"a}(2011)}]{JN2011}
Janhunen, T.; and Niemel{\"a}, I. 2011.
\newblock \emph{Compact Translations of Non-disjunctive Answer Set Programs to
  Propositional Clauses}, 111--130.

\bibitem[{Kabir et~al.(2022)Kabir, Everardo, Shukla, Hecher, Fichte, and
  Meel}]{KESHFM2022}
Kabir, M.; Everardo, F.~O.; Shukla, A.~K.; Hecher, M.; Fichte, J.~K.; and Meel,
  K.~S. 2022.
\newblock ApproxASP--a scalable approximate answer set counter.
\newblock In \emph{AAAI}, volume~36, 5755--5764.

\bibitem[{Kabir and Meel(2023)}]{KM2023}
Kabir, M.; and Meel, K.~S. 2023.
\newblock A Fast and Accurate ASP Counting Based Network Reliability Estimator.
\newblock In \emph{LPAR}, volume~94, 270--287.

\bibitem[{Korhonen and J{\"a}rvisalo(2021)}]{KJ2021}
Korhonen, T.; and J{\"a}rvisalo, M. 2021.
\newblock Integrating tree decompositions into decision heuristics of
  propositional model counters.
\newblock In \emph{CP}.

\bibitem[{Lagniez and Marquis(2017)}]{LM2017}
Lagniez, J.-M.; and Marquis, P. 2017.
\newblock An Improved Decision-DNNF Compiler.
\newblock In \emph{IJCAI}, volume~17, 667--673.

\bibitem[{Lee and Lifschitz(2003)}]{LL2003}
Lee, J.; and Lifschitz, V. 2003.
\newblock Loop formulas for disjunctive logic programs.
\newblock In \emph{ICLP}, 451--465. Springer.

\bibitem[{Lifschitz(2010)}]{Lifschitz2010}
Lifschitz, V. 2010.
\newblock Thirteen definitions of a stable model.
\newblock \emph{Fields of logic and computation}, 488--503.

\bibitem[{Lifschitz and Razborov(2006)}]{LR2006}
Lifschitz, V.; and Razborov, A. 2006.
\newblock Why are there so many loop formulas?
\newblock \emph{ACM Transactions on Computational Logic (TOCL)}, 7(2):
  261--268.

\bibitem[{Lin and Zhao(2004)}]{LZ2004}
Lin, F.; and Zhao, Y. 2004.
\newblock ASSAT: computing answer sets of a logic program by SAT solvers.
\newblock \emph{Artificial Intelligence}, 157(1): 115 -- 137.

\bibitem[{Marek and Truszczy{\'n}ski(1999)}]{MT1999}
Marek, V.~W.; and Truszczy{\'n}ski, M. 1999.
\newblock Stable models and an alternative logic programming paradigm.
\newblock In \emph{The Logic Programming Paradigm}, 375--398. Springer.

\bibitem[{Nouman et~al.(2016)Nouman, Yalciner, Erdem, and Patoglu}]{NYEP2016}
Nouman, A.; Yalciner, I.~F.; Erdem, E.; and Patoglu, V. 2016.
\newblock Experimental evaluation of hybrid conditional planning for service
  robotics.
\newblock In \emph{International Symposium on Experimental Robotics}, 692--702.
  Springer.

\bibitem[{Samer and Szeider(2007)}]{SS07}
Samer, M.; and Szeider, S. 2007.
\newblock Algorithms for propositional model counting.
\newblock In \emph{LPAR}, 484--498. Springer.

\bibitem[{Sharma et~al.(2019)Sharma, Roy, Soos, and Meel}]{SRSM2019}
Sharma, S.; Roy, S.; Soos, M.; and Meel, K.~S. 2019.
\newblock GANAK: A Scalable Probabilistic Exact Model Counter.
\newblock In \emph{IJCAI}, volume~19, 1169--1176.

\bibitem[{Thurley(2006)}]{Thurley2006}
Thurley, M. 2006.
\newblock sharpSAT--counting models with advanced component caching and
  implicit BCP.
\newblock In \emph{SAT}, 424--429. Springer.

\bibitem[{Tiihonen et~al.(2003)Tiihonen, Soininen, Niemel{\"a}, and
  Sulonen}]{TSNS2003}
Tiihonen, J.; Soininen, T.; Niemel{\"a}, I.; and Sulonen, R. 2003.
\newblock A practical tool for mass-customising configurable products.
\newblock In \emph{ICED}.

\bibitem[{Valiant(1979)}]{valiant1979}
Valiant, L.~G. 1979.
\newblock The complexity of enumeration and reliability problems.
\newblock \emph{SIAM Journal on Computing}, 8(3): 410--421.

\end{thebibliography}
\end{document}